\documentclass[12pt,openany]{article}
\usepackage[margin=1.4in]{geometry}

\usepackage{tocloft}
\cftsetrmarg{7em} 

\usepackage{authblk}

\usepackage[printonlyused]{acronym}

\acrodef{bdb}[BDB]{Beyond Design Basis}
\acrodef{cap}[CAP]{Condition Adverse to Quality Program}
\acrodef{cba}[CBA]{Cost-Benefit Analysis}
\acrodefplural{cba}[CBAs]{Cost-Benefit Analyses}
\acrodef{cdf}[CDF]{Core Damage Frequency}
\acrodef{chrs}[CHRS]{Containment Heat Removal System}
\acrodef{crmp}[CRMP]{Comprehensive Risk Management Process}
\acrodefplural{crmp}[CRMPs]{Comprehensive Risk Management Processes}
\acrodef{db}[DB]{Design-basis}
\acrodef{dba}[DBA]{Design Basis Accident}
\acrodef{did}[DID]{Defense-in-Depth}
\acrodef{eccs}[ECCS]{Emergency Core Cooling System}
\acrodef{esd}[ESD]{Event Sequence Diagram}
\acrodef{flc}[FC]{Fails Closed}
\acrodef{fmea}[FMEA]{Failure Modes and Effects Analysis}
\acrodefplural{fmea}[FMEA]{Failure Modes and Effects Analyses}
\acrodef{fsar}[FSAR]{Final Safety Analysis Report}
\acrodef{fs}[FoS]{Factor of Safety}
\acrodefplural{fs}[FsoS]{Factors of Safety}
\acrodef{ft}[FT]{Fault Tree}
\acrodef{fts}[FTS]{Fail To Start}
\acrodef{ftr}[FTR]{Fail To Run}
\acrodef{gra}[GRA]{Generation Risk Assessment}
\acrodef{lerf}[LERF]{Large Early Release Frequency}
\acrodef{mss}[MSS]{Main Steam System}
\acrodef{mslb}[MSLB]{Main Steam Line Break}
\acrodef{nei}[NEI]{Nuclear Energy Institute} 
\acrodef{npv}[NPV]{Net Present Value}
\acrodef{npp}[NPP]{Nuclear Power Plant}
\acrodef{nrc}[NRC]{Nuclear Regulatory Commission}
\acrodef{oqap}[OQAP]{Operations Quality Assurance Program}
\acrodef{lar}[LAR]{License Amendment Requests}
\acrodef{lb}[LB]{Licensing Basis}
\acrodef{lerf}[LERF]{Large Early Release Frequency}
\acrodef{loca}[LOCA]{Loss of Coolant Accident}
\acrodef{lwr}[LWR]{Light Water Reactor}
\acrodef{om}[O\&M]{Operations and Maintenance}
\acrodef{ora}[ORA]{Organizational Risk Assessment}
\acrodef{pga}[PGA]{Peak Ground Acceleration}
\acrodef{pra}[\emph{PRA}]{Probabilistic Risk Assessment}
\acrodef{pwr}[PWR]{Pressurized Water Reactor}
\acrodef{rcb}[RCB]{Reactor Containment Building}
\acrodef{rcd}[RCD]{Reactor Core Damage}
\acrodef{rcs}[RCS]{Reactor Coolant System}
\acrodef{rr}[RR]{Radiation Release}
\acrodef{ssc}[SSC]{Systems, Structures, and Components}
\acrodef{stm}[STM]{State Transition Matrix}
\acrodefplural{stm}[STMs]{State Transition Matrices}
\acrodef{ufsar}[UFSAR]{Updated Final Safety Analysis Report}
\usepackage[font={small},labelfont=bf]{subcaption}
\usepackage[font={small},labelfont=bf]{caption} 

\usepackage{framed}

\usepackage{amsfonts}
\usepackage{amsmath,amssymb,gensymb}
\usepackage{amsthm}
\usepackage{epstopdf}
\usepackage{cancel}

\usepackage{graphicx,xcolor}
\usepackage{subfloat}

\usepackage{booktabs}
\usepackage{multirow}
\usepackage{threeparttablex}
\usepackage{longtable}

\usepackage{enumitem}
\newlist{longenum}{enumerate}{5}
\setlist[longenum,1]{label=\Roman*)}
\setlist[longenum,2]{label=\Alph*)}
\setlist[longenum,3]{label=\roman*)}
\setlist[longenum,4]{label=\alph*)}
\setlist[longenum,5]{label=(\alph*)}
\setlist[longenum,6]{label=(\roman*)}

\title{%
Characterizing the Probability Law on Time Until Core Damage With \emph{PRA}:\\
\begin{large}
Consequences of Assuming Poisson Initiating Event Processes
\end{large} %
}
\author[]{Martin Wortman, Ernest Kee, \& Paul Nelson}
\affil[]{Texas A\&M University}

\usepackage{datetime}
\date{\today : \currenttime}                                           

\usepackage{comment}

\PassOptionsToPackage{hyphens}{url}\usepackage[hidelinks,pagebackref]{hyperref}
\usepackage{footnotebackref}

\usepackage{cleveref}

\newtheorem*{proposition*}{Proposition}
\newtheorem*{remark*}{Remark}

\usepackage{setspace}

\usepackage{natbib}
\begin{document}

\maketitle

\section*{Summary}

	Certain modeling assumptions underlying \ac{pra} allow a simple computation of core damage frequency (\emph{CDF}).  These assumptions also guarantee that the time remaining until a core damage event follows an exponential distribution having parameter value equal to that computed for the \emph{CDF}. While it is commonly understood that these modeling assumptions lead to an approximate characterization of uncertainty, we offer a simple argument that explains why the resulting exponential time until core damage distribution under--estimates risk.  Our explanation will first review operational physics properties of hazard functions, and then offer a non--measure--theoretic argument to reveal the the consequences of these properties for \emph{PRA}.\footnote{The conclusions offered, here, hold for any possible operating history that respects the underlying assumptions of \emph{PRA}.  Hence, the measure--theoretic constructs on filtered probability spaces is unnecessary for our developments.} We will then conclude with a brief discussion that connects intuition with our analytical development.
	
	\pagebreak
	
\section{Risk and Rational Hazard}
	A nuclear reactor can experience at most one core damage event. This catastrophic terminus is the central outcome with respect to which we shall characterize risk. In particular, we recognize the probability law on the time until core damage as a finite--characterization of risk.  To this end, \emph{Probabilistic Risk Assessment} (\emph{PRA}) offers a means for computing the probability distribution on the time until core damage. The \emph{PRA} methodology relies on certain assumptions that may or may not agree with physics characterizing reactor operations. When these assumptions are not in harmony with the underlying operational physics, the \emph{PRA} assessment of risk becomes (at best) approximate; the quality of the approximated probability law yielded by \emph{PRA} remains an open question.
	
	In the development that follows, we will identify physical and operational features that are generally shared by any maintained system (\emph{e.g., a nuclear reactor}) that can potentially succumb to a catastrophic failure. We then review the analytical assumptions underlying \emph{PRA} and rationalize these assumptions with the previously identified physical and operational features. We will then identify direct implications on risk (\emph{i.e.}, the probability distribution on time until core damage) when the assumptions underlying \emph{PRA} cannot be completely rationalized.  We will conclude that \emph{PRA} generally underestimates risk.  Finally, we cite some well--known results that help reveal the implicit difficulties in quantifying the quality of \emph{PRA} generated approximations of risk ... which remains a very important open research question.
\bigskip
	
\textbf{Physical and Operational Features of Maintained Systems:}
A maintained system can be described, ``\emph{Any system where restorative action is taken to lessen the likelihood of failure}." We shall characterize the likelihood of failure in the usual probabilistic sense as $R(t) \triangleq P(T>t)$, with the random variable $T$ being the time of first system failure after time $t=0$. Thus, we accept that $T$ is a continuous nonnegative random variable. We note that, the likelihood of failure is, thus, parametrically characterized by the hazard function $h(t)$, where

\begin{equation}
	R(t) = e^{-\int_0^t h(s)ds}
	\label{reliability}
\end{equation}
where,
 	$$h(t) \triangleq \lim_{s \downarrow 0}\frac{1}{s} P(T \le t+s | T >t).$$

It follows that by differentiating \eqref{reliability} that 
$$h(t) = \frac{-dR(t)}{R(t)}.$$

Consider \Cref{fig:hazard} showing a hypothetical hazard function behavior following maintenance interventions (epochs) that result in discontinuities; we do not require that $h(t)$ be continuous.\footnote{The hazard $h(t)$, \emph{i.e.,} the propensity for first system failure in the next instant of time, will surely be discontinuous with downward jumps at epochs of successful maintenance.}
Importantly, note that $h(t)$ and $R(t)$ carry exactly the same information (as one is uniquely determined by the other). This information includes a given (fixed) history of maintenance prior to the epoch of first system failure. For any given maintenance history, there will correspond a unique hazard trajectory $h(t)$, $t \ge 0.$ Further, operational physics requires that any fixed hazard trajectory $h(t)$ should have the following properties that we identify as \emph{principles of rational hazard}.
\bigskip

\textbf{Principles of Rational Hazard}
\begin{enumerate}
	\item $0 < h(t) < \infty$, $\forall t \ge0$; hazard is positive and finite.
	\item $\lim_{s \downarrow 0}h(t+s) = h(t)$, $\forall t,s \ge 0$; hazard is right--continuous.
	\item $\lim_{s \downarrow 0}\frac{h(t+s) - h(t)}{t - s} \ge 0$, $\forall t\ge0$; hazard is non--decreasing almost everywhere.
   \item $\lim_{s \downarrow 0}h(t) - h(t-s) < 0$; hazard decreases only at each of the (at most countable number of) epochs of successful maintenance.
	\item $h(0) = \inf_{t \ge0} h(t)$; without loss of generality, the system is as \emph{good--as--new} at time zero.
	
\end{enumerate}
\begin{figure}[h]
	\centering
	\includegraphics[width=0.8\textwidth]{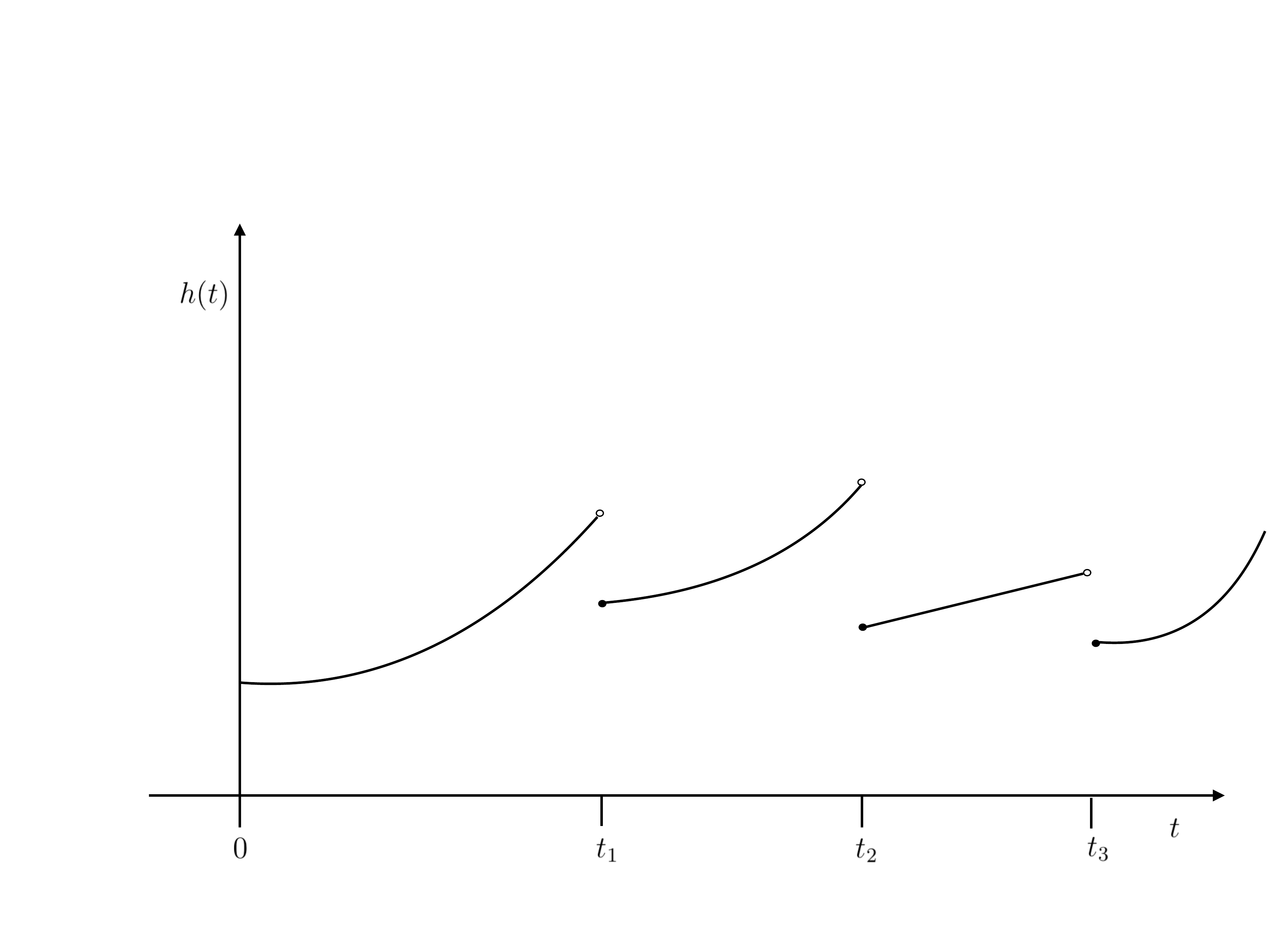}
	\caption{Trajectory of a hypothetical hazard for a fixed degradation and maintenance.}
	\label{fig:hazard}
\end{figure}

Returning, now, to the characterization of general maintained systems, we introduce the follow proposition. 

    \begin{proposition*}
    Maintained systems satisfying the principles of rational hazard have a time of first failure complementary distribution $R(t)$ that is stochastically ordered with respect to $e^{-h(0)t}$, $t \ge 0.$
    \end{proposition*}
    
    \begin{proof}
    By eq\eqref{reliability} the complimentary distribution of time of first failure is given by $R(t) = e^{-\int_0^t h(s)ds},$ $t \ge 0.$ Appealing to the Principles of Rational Hazard, it follows from principles 1 and 2 that $h(t)$ is integrable on $t \ge 0.$ It then follows from principles 3 and 4 that $\int_0^th(s)ds$ must be positive, finite, and monotone nondecreasing. Thus, $$\{\inf_{t \ge 0} h(t)\}t \le \int_0^t h(s)ds,$$ and by principle 5, $$h(0)t \le \int_0^t h(s)ds.$$  Exponentiating now shows that for every possible trajectory of rational hazard
    $$e^{-h(0)t} \ge e^{-\int_0^t h(s)ds},$$ or equivalently, $$P(T \le t) \ge 1 - e^{-h(0)t}, \forall t \ge 0.$$
    \end{proof}

    We now observe that since time $t=0$ can, without loss of generality, be set to any maintenance epoch that returns the system to a \emph{good--as--new} condition, maintained systems that conform to the principles of rational hazard will have a time of first failure distribution (for all possible maintenance trajectories) that is bounded above by an exponential distribution with parameter equal to the initial hazard.
    
    \begin{remark*}
    The above proposition, unto itself, is not especially noteworthy, since intuition easily suggests exactly what proof formalizes: Disregarding the effects of degradation on a maintained system will surely lead to under--estimating failure risk. However, we can leverage this proposition to obtain deeper insights regarding the consequences of Poisson assumptions that underlie PRA. 
    \end{remark*}

\section{Approximations Using \emph{PRA} } 
The general stochastic processes governing a maintained system's temporal behavior prior to core damage are so complicated as to go beyond practical computational analysis.  Hence, simplified models are employed that yield approximate characterizations of core damage risk: \emph{PRA} is the most commonly used methodology to quantify core damage risk.
Essential to the \emph{PRA} methodology is that the first (and only) core damage events arrives according to a Poisson process.\footnote{Although other reactors at large commercial multi-unit sites may continue to operate (for example, the Three Mile Island site experience and the Chernobyl experience), the reactor(s) involved in the accident is(are) have historically been decommissioned. An exception may be Fukushima where three reactors avoided core damage but have not yet been restarted.}  An analytical consequence of Poisson arrivals is that the time until the first (and only) core damage event follows an exponential distribution, and the parameter of this distribution takes a value identical to what is termed core damage frequency (\emph{CDF}).\footnote{Clearly, \emph{CDF} does not imply that there exists multiple core damage events; the (assumed Poisson) arrival process of core damage terminates upon a damage epoch.  Thus, it is well understood that \emph{CDF} simply refers to expected number of arrivals per-unit time of the un-terminated Poisson model characterizing the time of the first (and only) core damage event.} A specific value of {CDF} is computed using the \emph{PRA} predictive modeling methodology and noting that

    $$h = CDF = \frac{1}{E[T]},$$
we have that

\begin{equation}
	R(t) = P(T>t) = e^{-ht}, t \ge 0.
	\label{exponential} 
\end{equation}

Clearly, when core damage is approximated as a Poisson arrival, the practical effects of degradation and maintenance are not represented.  That is, Poisson arrivals assume that the system hazard remains constant over all time.\footnote{Note that the constant hazard, characteristic of eq\eqref{exponential}, satisfies each of the five principles of rational hazard.} Hence, by \emph{PRA} approximation, a system remains in the \emph{good--as--new} condition until the occurrence of core damage.

Practitioners, of course, recognize that the constant hazard trajectories characteristic of \emph{PRA} are inconsistent with reality. Thus, the quality of core damage risk characterization of eq\eqref{exponential} is of central importance.  To this end, we note that all stochastic approximations can be classified according to one of two properties: 1) the approximation provides a bound, or 2) the approximation does not provide a bound.  Bounding approximations are typically ordinal, giving a means to gauge true system performance as being ``at least as good as" or ``never any worse than" characterizations.  Non-bounding approximations typically rely on a distance metric, where system performance can be gauged as within some quantifiable (computable) distance from the approximation.

Under the principles of rational hazard and the concomitant stochastic ordering proposition, the \emph{PRA} approximation of eq\eqref{exponential} represents an uninformative lower bound on core damage risk. That is: ``The risk of core damage by time~$t$ will always be higher than that predicted by the \emph{PRA} approximation."\footnote{The uninformative lower bound on core damage risk is an obvious consequence of Poisson modeling, where the effects of degradation and maintenance are not captured.}

In order to improve the quality of the \emph{PRA} approximation one could seek to relax one or more of the rational hazard principles with the objective of finding either a better bound or a quantifiable distance metric between the \emph{PRA} approximation and stochastic system behavior. However, it is unlikely that the \emph{PRA} approximation can be used to obtain a better bound or quantify distance to stochastic system behavior without a more detailed characterization of degradation and maintenance. Including degradation and maintenance information (\emph{e.g.,} histories of equipment failure and repair) reveals the difficulties in modeling hazard processes which are generally represented as the stochastic intensity of certain underlying marked--point processes. 

Finally, we note that there exists a large and established literature reporting results on approximating general stochastic point processes by Poisson processes. This line of investigation largely centers on ``Stein's Method" (also referred to as the ``Stein--Chen" method) to develop bounds on distance metrics characterizing how closely a Poisson process approximates a stationary point process. The works of \cite{Barbour1989}, \cite{Chen2004}, and \cite{Chatterjee2005} and their references provide a fairly complete coverage of the state--of--the--art in Poisson approximations. While these bounded distance metrics provide useful devices in proving certain contraction properties associated with weak--convergence arguments for limit theorems, none lends itself to direct quantification. Only incremental progress along this line of investigation has been made in recent years.

\bibliography{PRA.bib}

\begin{thebibliography}{}

\bibitem[\protect\citeauthoryear{Barbour and Holst}{Barbour and
  Holst}{1989}]{Barbour1989}
Barbour, A.~D. and L.~Holst (1989).
\newblock {Some applications of the Stein-Chen method for proving Poisson
  convergence}.
\newblock {\em Adv. Appl. Probab.\/}~{\em 21\/}(1), 74--90.

\bibitem[\protect\citeauthoryear{Chatterjee, Diaconis, and Meckes}{Chatterjee
  et~al.}{2005}]{Chatterjee2005}
Chatterjee, S., P.~Diaconis, and E.~Meckes (2005).
\newblock {Exchangeable pairs and Poisson approximation}.
\newblock {\em Probab. Surv.\/}.

\bibitem[\protect\citeauthoryear{Chen and Xia}{Chen and Xia}{2004}]{Chen2004}
Chen, L.~H. and A.~Xia (2004).
\newblock {Stein's method, Palm theory and Poisson process approximation}.
\newblock {\em Ann. Probab.\/}.

\end{thebibliography}
\end{document}